\documentclass[letterpaper,11pt]{amsart}
\usepackage{tikz-cd}

\usepackage{latexsym,array,delarray,amsthm,amssymb,epsfig}
\usepackage{caption}
\usepackage{subcaption}
\usepackage{amscd}

\theoremstyle{plain}
\newtheorem{thm}{Theorem}[section]

\newtheorem{op}[thm]{Proposition}
\newtheorem{cor}[thm]{Corollary}

\newtheorem*{thm*}{Theorem}
\newtheorem*{lemma*}{Lemma}
\newtheorem*{prop*}{Proposition}
\newtheorem*{cor*}{Corollary}
\newtheorem*{conj*}{Conjecture}

\theoremstyle{definition}

\newtheorem{ex}[thm]{Example}

\theoremstyle{remark}

\newcommand{\ind}{\mbox{$\perp \kern-5.5pt \perp$}}

\begin{document}

\title[Identifiability and Reconstructibility Under a Modified Coalescent]
{Identifiability and Reconstructibility of Species Phylogenies Under a Modified Coalescent}
\author{Colby Long and Laura Kubatko}

\begin{abstract}
	Coalescent models of evolution account for incomplete 
	lineage sorting by specifying a species tree parameter which 
	determines a distribution on gene trees. 
	It has been shown that the unrooted topology of the species
	tree parameter of the multispecies coalescent is 
	generically identifiable. Moreover, a statistically consistent 
	reconstruction method called SVDQuartets has been 
	developed to recover this parameter. In this paper,
	we describe a modified 
	multispecies coalescent model that allows
	for varying effective population size and violations of the 
	molecular clock. We show that the unrooted topology 
	of the species tree for these models is generically
	identifiable and that SVDQuartets is still a statistically 
	consistent method for inferring this parameter.
\end{abstract}

\maketitle

%
%
%

\section{The multispecies coalescent}

The goal of phylogenetics is to reconstruct the evolutionary history
of a group of species from biological data. Most often, the data available are the aligned DNA sequences of the species under consideration. The descent of these species from a common 
ancestor is represented by a phylogenetic tree which we call
the \emph{species tree}. However, it is well known that due to various biological phenomena, such as horizontal gene transfer and incomplete lineage sorting, the ancestry of individual genes will not necessarily match the tree of the species in which they reside \cite{pamilonei1988,syvanen1994,maddison1997}. There are various phylogenetic reconstruction 
methods that account for this discrepancy in different ways.
One approach is to reconstruct individual 
\emph{gene trees} by some method and them utilize this
information to infer the original species tree \cite{liuetal2009b,liuetal2010,wu2012,mirarabetal2014,mirarabwarnow2015}. 

The multispecies coalescent model incorporates incomplete lineage sorting directly. The tree parameter of the model is the species tree, an $n$-leaf rooted equidistant tree with branch lengths. The species tree yields a distribution on possible
gene trees along which evolution is modeled by a 
$\kappa$-state 
substitution model. 
For a fixed choice of parameters,
the multispecies coalescent returns
a probability distribution on the $\kappa^n$ possible 
$n$-tuples of states that may be observed.
In order to infer the species tree from data, one searches
for model parameters yielding a distribution close to that 
observed, using, for example, maximum likelihood.

In \cite{chifmankubatko2015}, the authors show that
given a probability distribution from the multispecies coalescent model it is possible to infer the topology of the unrooted species tree parameter. 
Restricting the species tree to any 4-element subset of the leaves 
yields an unrooted 4-leaf binary phylogenetic tree called a  \emph{quartet}. 
For a given label set, there are only
three possible quartets which each induce
a flattening of the probability tensor.
Given a probability distribution arising from the multispecies 
coalescent, the flattening matrix corresponding to the quartet
compatible with the species tree will be rank 
${\kappa + 1 \choose 2}$ or less while
the other two will generically have rank strictly
 greater than this value.
Since the topology of an unrooted tree is uniquely determined by
quartets \cite{Semple2003}, these flattening matrices can be used to determine the unrooted topology of the species tree exactly. 
Of course, empirical and even simulated data produced by the multispecies coalescent will only approximate the distribution of the model.
Therefore, the same authors also proposed a method called SVDQuartets \cite{chifmankubatko2014}, which uses singular value decomposition to infer each quartet topology by determining
 which of the flattening matrices is closest to the set of rank ${\kappa + 1 \choose 2}$ matrices.

The method of SVDQuartets offers several advantages over 
other existing phylogenetic reconstruction methods. 
It is statistically consistent, accounts for 
incomplete lineage sorting, and is computationally
much less expensive than Bayesian methods achieving
the same level of accuracy.
It is often under-appreciated that this reconstruction method
can be used to recover the species tree for several
different underlying nucleotide substitution models
without any modifications.
It was shown in \cite{chifmankubatko2015} that the method of SVDQuartets is statistically consistent if the underlying 
model for the evolution of sequence data 
along the gene trees is the 4-state 
general time-reversible (GTR) model
or any of the commonly used submodels thereof 
(e.g., JC69, K2P, K3P, F81, HKY85, TN93).
Thus, the method
does not require any a priori assumptions about the 
underlying nucleotide substitution process other than time-reversibility. 

In this paper, we show that the method of SVDQuartets 
has more theoretical robustness even 
than has already been shown.
We will specifically focus on the case where
the underlying nucleotide substitution model is 
one of the the 4-state models most
widely used in phylogenetics.
We describe several modifications to the 
classical multispecies coalescent model 
to allow for more realistic 
mechanisms of evolution.
For example, we remove the assumption of a 
molecular clock by removing the restriction that the 
species tree be equidistant. 
We also allow the effective population size to vary on each branch
and show how these two relaxations together can 
be interpreted as a scaling of certain rate matrices in the 
gene trees. 
Remarkably, we show that the unrooted species
tree of these modified models is still identifiable and 
that  SVDQuartets is a 
statistically consistent reconstruction 
method for these modified models.
Thus, despite the introduction of several parameters,
 effective and efficient methods for reconstructing
 the species tree of these modified coalescent models are
already available off the shelf and implemented in PAUP$^*$ \cite{Swofford2002}.

In Section \ref{sec: Multi-species Coalescent}, we review the 
classical multispecies coalescent model and 
discuss some of its limitations in modeling certain
biological phenomena. 
We then describe several modifications to the
classical model to remedy these weaknesses.
In Section \ref{sec: Identifiability}, we establish the theoretical properties of identifiability for these families of modified coalescent 
models. Finally, in Section \ref{sec: Conclusions}, we describe why SVDQuartets is a strong candidate for reconstructing the species tree of these models and propose several other modifications that we might make to the multispecies coalescent.

\section{The multispecies Coalescent}
\label{sec: Multi-species Coalescent}

\subsection{Coalescent Models of Evolution}
\label{subsec: Coalescent Models}

In this section, we briefly review the multispecies coalescent
model and explain how the model yields a probability 
distribution on nucleotide site patterns. 
As our main results will parallel those 
found in \cite{chifmankubatko2015}, we will import much of the
notation from that paper and refer the reader there for
a more thorough description of the model.

The Wright-Fisher model from population genetics models the convergence of multiple lineages backwards in time toward a common ancestor.
Beginning with $j$ lineages from the current generation,
the model assumes discrete generations with 
constant effective population size $N$.
In each generation, each lineage is assigned a parent
uniformly from the previous generation. 
For diploid species, there are $2N$ copies of each gene
in each generation, and thus the probability of selecting any 
particular gene as a parent is $\frac{1}{2N}$. 
Two lineages are said to coalesce when they select the same parent in a particular generation.

As an example, if we begin with two lineages in the same species, the probability they have the same parent in the previous generation,
and hence coalesce,
 is $\frac{1}{2N}$ and the probability that they do not coalesce 
 in this generation is 
 $(1 - \frac{1}{2N})$. 
Therefore, the probability that two lineages coalesce in exactly 
the $i$th previous generation is given by 
$$\left ( \frac{1}{2N} \right )  \left (1 - \frac{1}{2N} \right )^{i-1}.$$
For large $N$, the time at which the two lineages coalesce, $t$, follows an exponential distribution with rate $(-2N)^{-1}$, where time is measured in number of generations.
Every $2N$ generations is called a \emph{coalescent unit} and
time is typically measured in these units to simplify
the formulas for time to coalescence.
However, in this paper, we will introduce separate effective
population size parameters for each branch of the species trees.
So that our time scale is consistent across the tree we will work 
in generations rather than coalescent units.
%
%
In these units, for $j$ lineages, the time to the next coalescent event $t_j$ has probability density,



$$f(t_j) = \dfrac{j(j-1)}{2}
\left (
\dfrac{1}{2N}
\right )
\exp \left ( - \dfrac{j(j-1)}{2}
\left (
\dfrac{1}{2N}
\right )
t_j \right ), \ t_j > 0.$$

\noindent This is typically referred to as {\itshape Kingman's coalescent} \cite{kingman1982a,kingman1982b,kingman1982c,tajima1983,tavare1984,takahata1985}.

The multispecies coalescent is based on the same framework, 
but we assume that the species tree of the $n$ sampled taxa
is known. We let $S$ denote the topology (without branch lengths) of the $n$-leaf rooted binary phylogenetic species tree.
We then assign a positive edge weight to each edge of 
$S$. The weights are chosen to given an equidistant edge
weighting, meaning that
the sum of the weights of the branches on the path between 
any leaf and the root is the same. To specify an equidistant
edge-weighting, it is enough to specify the vector of
speciation times $\tau = (\tau_1,\ldots, \tau_{n-1})$.
The weights give the length of each branch and so
we let $(S,\boldsymbol{\tau})$ denote the species tree with branch lengths.
The condition that $S$ be equidistant implies that the
expected number of mutations along the path from the root
to each leaf is the same. This is what is known as
as imposing a molecular clock.

Once this species tree is fixed,
the multispecies coalescent gives a probability density
on possible gene trees.  
All of the same assumptions as before apply, except that it is 
now impossible for two lineages to coalesce if they are 
not part of the same population. Hence, lineages
may only coalesce if they are in the same branch of $S$.
%
%
%
%
%
We denote the topology of
a rooted $n$-leaf binary phylogenetic gene tree by $G$ and
let $\mathbf{t} = (t_1, \ldots, t_{n-1})$ denote the vector of coalescent
times. As for the species tree, $(G,\mathbf{t})$ 
then represents a gene
tree with weights corresponding to
branch lengths.
A coalescent history $h$ refers to a 
particular sequence of coalescent events
as well as the branches along which they occur.
Note that there are only
finitely many possible \emph{coalescent histories} compatible
with $S$
and we call the set of all such histories $\mathcal{H}$.
For a particular history, 
we can compute the probability density for 
gene trees with that history explicitly. 
We write the gene tree density of gene trees conforming
to a particular
coalescent history $h$ in $S$ as
$f_h((G,\mathbf{t})|(S,\boldsymbol{\tau}))$. 

\begin{ex}
\label{ex: SpeciesTree}
Let $S$ be the 4-leaf species tree depicted in  
Figure \ref{fig: SpeciesTree} and let $h$ refer to the coalescent history in which 
\begin{enumerate}
\item Lineages $a$ and $b$ coalesce in $e_{ab}$.
\item Lineages $c$ and $d$ coalesce in $e_{abcd}$.
\item Lineages $ab$ and $cd$ coalesce in $e_{abcd}$.
\end{enumerate}
The gene tree depicted inside of $S$ is compatible with $h$.

\begin{figure}[h]
\caption{A species tree with a compatible gene tree. Speciation times, denoted as $\tau_{X}$, indicate the time from present back to the point at which all species in the set $X$ are part of the same ancestral species.  The coalescent times $t_j$ denote the time back from the most recent speciation event to a coalescent event.}
\label{fig: SpeciesTree}
\begin{center}
\includegraphics[width=7.5cm]{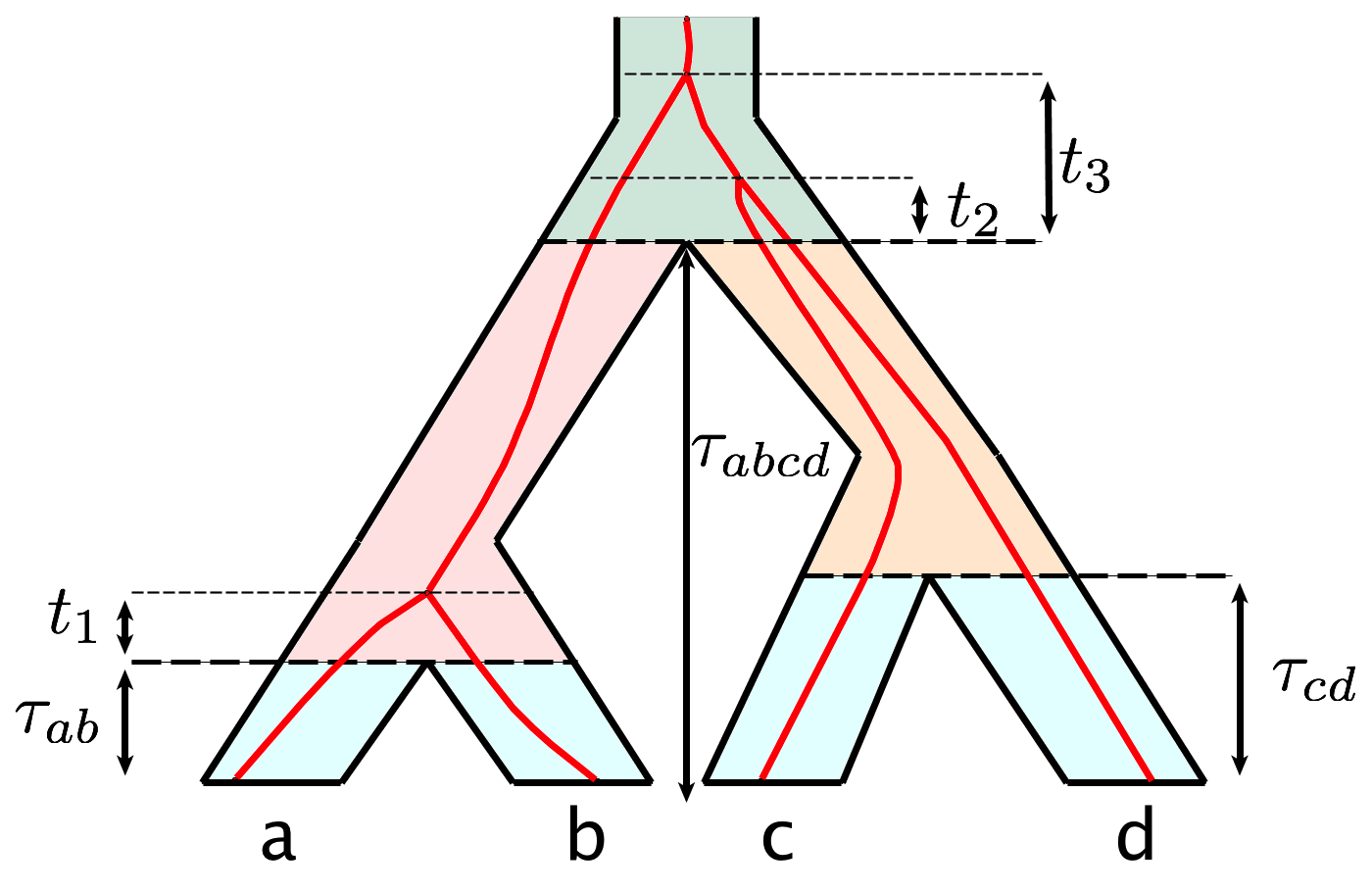}
\end{center}
\end{figure}

The probability of observing a gene tree with history $h$ for the
species tree $S$ under the multispecies coalescent model is
$\int_\textbf{t} \! f_h((G,\mathbf{t})|(S,\boldsymbol{\tau})) d\textbf{t} =$
\medskip
$$
\int_0^{\infty} \!
\int_0^{t_3} \!
\int_0^{\tau_{abcd} - \tau_{ab}} \!
\left(
\dfrac{1}{2N}
\right )^3
\exp \left ( 
\dfrac{-t_1}{2N}
\right )
\exp 
\left ( 
\dfrac{-(\tau_{abcd} - \tau_{cd})}{2N}
\right )
%
\exp \left ( 
\dfrac{-3t_2}{2N}
\right )
%
\exp \left ( 
\dfrac{-t_3}{2N}
\right )
\ 
dt_1
dt_2
dt_3.
$$
\medskip

The term involving $(\tau_{abcd} - \tau_{cd})$ is the probability 
that lineages $c$ and $d$ do not coalesce in the branch $e_{cd}$.
\end{ex}

Once the gene tree for a locus is specified, we model the evolution 
along this gene tree as a continuous-time homogenous Markov 
process according to a nucleotide substitution model. 
The model gives a probability distribution on the set
of all $4^n$ possible $n$-tuples of observed states 
at the leaves of $(G,\mathbf{t})$. 
We can write the probability of observing the 
state $(i_1,\ldots,i_n)$ 
as 
$p^*_{i_1\ldots i_n|(G,\mathbf{t})}$. 
Precisely how this distribution is
calculated is described in \cite{chifmankubatko2015}. 
Here, we sketch the relevant details needed to introduce the 
modified multispecies coalescent model described in the 
next section.

For a $4$-state substitution model, 
there is a $4 \times 4$ instantaneous rate matrix $Q$
where the entry $Q_{ij}$ encodes the rate of 
conversion from state $i$ to state $j$. 
The distribution of states at the root
is $\boldsymbol{\pi} = (\pi_T, \pi_C,\pi_A, \pi_G)$ where 
$\boldsymbol{\pi}$ is the
stationary distribution of the rate matrix $Q$. 
To compute the probability of observing a particular state at the leaves, we associate to each vertex $v$ a random variable 
$X_v$ with state space equal to the set of $4$ possible states. To the root, we associate the random variable $X_r$ where $P(X_r = i) = \pi_i$.
Letting $t_{e}$ be the length of edge $e = uv$, 
$P(t_e) = e^{Qt_{e}}$ is the matrix of transition probabilities along that edge. That is, $P_{ij}(t_e) = P(X_v = j | X_u = i)$.
Given an assignment of states to each vertex of the tree we 
can compute the probability of observing this state using
the Markov property and the appropriate entries of the transition matrices. To determine the probability of observing a particular state at the leaves, we marginalize over all possible states of the
internal nodes.

In this paper, we are primarily interested in 
4-state models of DNA evolution where the 
four states correspond to the DNA bases.
Different phylogenetic models place different
restrictions on the entries of the rate matrices.
The results that we prove in the next section will apply
when the underlying nucleotide substitution model is any
of the commonly used 4-state time-reversible models.
We note here that because these models are 
time-reversible the location of the root in each gene tree
is unidentifiable from the distribution for each gene tree
\cite{felsenstein1981}.
In subsequent sections, we will introduce and describe similar results for the JC+I+$\Gamma$ model, that allows for invariant sites and gamma-distributed rates across sites. 

\begin{figure*}[h]
    \centering
    \begin{subfigure}[t]{0.5\textwidth}
        \centering
      $\begin{pmatrix}
        * & a & b & c  \\
        a & * & c & b  \\
        b & c & * & a  \\
        c & b & a & *  \\
        \end{pmatrix}$
         \caption{Kimura 3-parameter model (K3P)}
    \end{subfigure}%
    ~ 
    \begin{subfigure}[t]{0.5\textwidth}
        \centering
           $\begin{pmatrix}
        * & \alpha\pi_C &  \beta \pi_A& \gamma \pi_G \\
        \alpha\pi_T & * &  \delta \pi_A& \epsilon \pi_G \\
        \beta \pi_T & \delta\pi_C &  * & \eta \pi_G \\
        \gamma\pi_T & * &  \eta \pi_A& * \\
        \end{pmatrix}$
          \caption{4-state general time-reversible model (GTR)}
    \end{subfigure}
    \caption{The rate matrices for two commonly used models in phylogenetics. The diagonal entries are chosen so that the row sums are equal to zero. In the K3P model, the root distribution is
 assumed to be uniform.}
 \label{fig: rate matrices}
\end{figure*}

Now, given a species tree $(S,\boldsymbol{\tau})$ and a choice
of nucleotide substitution model, let 
$p_{i_1\ldots i_n|(S,\boldsymbol{\tau})}$ be the probability of observing 
the site pattern $i_1\ldots i_n$ at the leaves of $(S,\boldsymbol{\tau})$.
To account for gene trees with topologies different from
that of $S$, we label each leaf of a gene tree by 
the corresponding label of the leaf branch of $S$ 
in which it is contained.
To compute $p_{i_1\ldots i_n|(S,\boldsymbol{\tau})}$, we must
consider the contribution of each gene tree compatible with 
$(S,\boldsymbol{\tau})$. So that we may write the formulas explicitly, 
we first consider the contribution of gene trees matching a particular coalescent history,

$$
p_{i_1\ldots i_n|h,(S,\boldsymbol{\tau})}=
\displaystyle
\int_{\mathbf{t}}
p^*_{i_1\ldots i_n|(G,\mathbf{t})}
f_h((G,\mathbf{t})|(S,\boldsymbol{\tau}))
\ d\mathbf{t}
$$

To compute the total probability, we sum over all possible histories so that

\begin{align*}
p_{i_1\ldots i_n|(S,\boldsymbol{\tau})} &=
\displaystyle
\sum_{h \in \mathcal{H}}
\displaystyle
\int_{\mathbf{t}}
p^*_{i_1\ldots i_n|(G,\mathbf{t})}
f_h((G,\mathbf{t})|(S,\boldsymbol{\tau}))
\ d\mathbf{t} \\
&= 
\displaystyle
\sum_{h \in \mathcal{H}}
p_{i_1\ldots i_n|h,(S,\boldsymbol{\tau})}
\end{align*}

Note that the bounds of integration will depend on the history
being considered.

\subsection{A Modified Coalescent}

In this section, we introduce various ways that we might
alter the multispecies coalescent to  
better reflect the evolutionary process. 
Recall that the
length of the path from the root of the 
species tree to a leaf is the total
number of generations that have occurred between
the species at the root and that at the leaf.
Since the length of a generation
may vary for different species \cite{martinpalumbi1993}, 
it may be desirable to allow 
the lengths of the paths from the root to each leaf to differ.
Therefore, we first consider expanding the allowable set of
species tree parameters to include trees that are not equidistant. 
Such a tree cannot be specified by the speciation times alone, 
so we use the vector
 $\boldsymbol{\tau} = (\tau_1, \ldots, \tau_{2n -2})$ to denote
the vector of branch lengths. 

Fix a nucleotide substitution model. 
Let $\mathcal{C}(S) \subseteq \Delta^{4^n  - 1}$ be the set of probability distributions obtained from the standard multispecies 
coalescent model on the $n$-leaf topological tree $S$. 
Let $\mathcal{C}_n \subseteq \Delta^{4^n  - 1}$ denote the set
of all distributions obtained by allowing the species tree to vary
over all equidistant $n$-leaf trees. 
If the species tree of the model is not assumed to be equidistant,
this removes the assumption of a molecular clock and we 
refer to this model as the \emph{clockless coalescent}.
The set of distributions obtained from a single tree topology 
in the clockless coalescent is $\mathcal{C}^*(S)$ and the 
set of distributions obtained by varying the species tree over
all rooted
$n$-leaf trees is
$\mathcal{C}_n^*$.

We can also account for the fact that the effective
 population size, $N$, may vary for different species \cite{charlesworth2009} 
 by introducing a separate effective population size parameter 
 for each internal branch of the species tree.
 We call this model the \emph{$p$-coalescent} and
 let $\mathcal{C}(S,N)$ be the set of probability 
distributions obtained from the $p$-coalescent on $S$. 
In analogy to our notation from above, we use
$\mathcal{C}_n(N)$ to denote the set of all distributions
obtained from the $p$-coalescent and use
$\mathcal{C}^*(S,N)$ and $\mathcal{C}^*_n(N)$ for
the clockless $p$-coalescent.

Notice that since we assume coalescent events do not
occur within leaves, changing the effective population
size on the leaf edges does not change the gene tree distribution 
or the site pattern probability distribution from the gene trees. 
In this section, we will show that remarkably, the unrooted species tree parameter of the clockless coalescent, $p$-coalescent, and the clockless $p$-coalescent are all generically identifiable.
Conveniently, from the perspective of reconstruction,
we also show that the method of SVDQuartets \cite{chifmankubatko2014} 
can be used to reconstruct the unrooted species tree
based on a sample from the probability distribution given by the model.

It is worth pausing for a moment to compare the mathematical 
and biological meaning of these changes. 
To illustrate, we will consider
how the model on the species tree $S$ depicted in 
Figure \ref{fig: SpeciesTree} changes as we modify the parameters. 
First, consider the effect of increasing the length of the branch $e_{ab}$ in $S$.
This is essentially modeling an increase in the generation
rate of the common ancestor of 
species $a$ and $b$. 
Notice that this changes the gene tree distribution 
and that 
this is as it should be.
 Indeed, thinking in terms of the 
Wright-Fisher process,
an increase in the number of generations along this branch
should correspond to
an increase in the probability of two lineages coalescing
 in this branch. 
 
Similarly, we may consider the $p$-coalescent on the
equidistant tree $S$. There is now a parameter $N_{ab}$
representing the effective population size
of the immediate ancestor of species $a$ and $b$.
Altering $N_{ab}$ also changes the gene tree distribution, since  an increase or decrease of $N_{ab}$ will have the opposite effect on the probability that branches $a$ and $b$ coalescence along $e_{ab}$. 
However, we emphasize that there is no reason to assume that 
we can obtain the same distribution obtained by altering the
length of $e_{ab}$ by adjusting the value of $N_{ab}$ for $S$.
When we alter branch length, both the formulas for time
to coalescence and for the transition probabilities
change, whereas, when we alter the effective population
size on a branch, only the formulas for time 
to coalescence are changed.

Finally, rather than altering branch lengths, one might think to model the increase in the number of generations along
a branch by scaling the rate matrix $Q$ by some factor
 $\rho$ along $e_{ab}$. 
To compute the transition probabilities for a specific branch of a gene tree, we would then use the rate matrices for the branches of the species tree in which it is contained. Thus, multiple transition matrices may have to be multiplied to compute the transition matrix on a single branch of a gene tree. 
For example, if we scale the rate matrix for $e_{ab}$, the 
transition matrix along the leaf edge labeled by $a$ in the
gene tree depicted in Figure \ref{fig: SpeciesTree} 
will be $e^{Q(\tau_{ab})}e^{(\rho Q)t_1} = e^{Q(\tau_{ab}+ \rho t_1)}$.
This idea also allows us to consider 
a multispecies coalescent model with an underlying heterogenous Markov process which we explore
in a forthcoming paper.
 
Since the transition probabilities are given by entries of a matrix
 exponential, and 
 $\exp((\rho Q)\tau) = \exp(Q(\rho \tau))$, 
 it would seem that scaling the rate matrix for a single branch
 or scaling the length of that branch give the same model.
However, while this is true for a nucleotide substitution model, it is not necessarily true for the multispecies coalescent since scaling the rate matrix does not change the gene tree distribution, but 
changing the length of the branch does. 
Still, it may be desirable to scale a rate matrix to model an
actual change in the per generation mutation rate. 
However, there is actually no need to introduce scaled rate matrices, since,
as Example \ref{ex: scaling} demonstrates, 
the clockless $p$-coalescent 
effectively models this situation as well.
In fact, we can obtain the same model by fixing any one of 
the branch length (e.g., set $\boldsymbol{\tau} = \mathbf{1}$), effective population size, or rate matrix scaling factor in the species
tree and introducing distinct parameters on each branch 
for the other two. 
We have chosen to introduce parameters for effective population size and branch length since these have the most natural biological interpretations.


\begin{figure}
\centering
  \includegraphics[width=4cm]{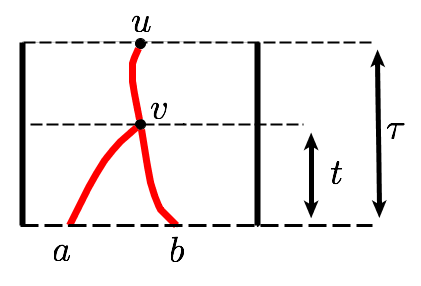}
  \captionof{figure}{Two lineages coalescing in a branch of a species tree.}
  \label{fig: scaling}
\end{figure}

\begin{ex}
\label{ex: scaling}
Let $a$ and $b$ be two lineages entering a 
branch $e$ of a species tree as in Figure \ref{fig: scaling}.
Let $\tau$ be the length of this branch and $N$ be the effective population size parameter. 
The probability that $a$ and $b$ do not coalesce in $e$ is 
\begin{align}
1 - \int_0^\tau  \!
\left(
\dfrac{1}{2N}
\right )
\exp \left ( 
\dfrac{-t}{2N}
\right ) \ dt
= 
\exp \left ( 
\dfrac{-\tau}{2N}
\right ).
\end{align}
\end{ex}

If $a$ and $b$ coalesce, then we can compute the probability of observing the state $xy$ at $a$ and $b$ under a homogenous Markov model model 
where the rate matrix on the branch $e$ is scaled by a factor $\rho$. 
We assume the distribution of states at the vertex $u$ is the 
vector $\boldsymbol{\pi}$. 
Thus, we have,

\begin{align}
 p_{xy} = 
\displaystyle \sum_{z_1,z_2} 
\displaystyle \int_0^{\tau} \!
\left(
\dfrac{1}{2N}
\right )
\exp \left ( 
\dfrac{-t}{2N}
\right )
\pi_{z_1}
\exp(\rho Q(\tau - t))_{z_1,z_2}
\exp(\rho Q(t))_{z_2,x}
\exp(\rho Q(t))_{z_2,y}
\ dt.
\end{align}

Consider now if instead of scaling the rate matrix $Q$ by $\rho$, we scale the length of $e$ and the effective population size
by $\rho$. Then the probability that $a$ and $b$ do not coalesce
remains unchanged since 
$$
\exp \left ( 
\dfrac{-\rho\tau}{2\rho N}
\right )
=
\exp \left ( 
\dfrac{-\tau}{2N}
\right ).$$

Likewise, the probability of observing state $xy$ is given by the following formula, where we make the substitution  $t = \rho T$.

\begin{align*}
 p_{xy} &= 
  \displaystyle \sum_{z_1,z_2} 
\displaystyle \int_0^{\rho \tau} \!
\left(
\dfrac{1}{2\rho N}
\right )
\exp \left ( 
\dfrac{-t}{2\rho N}
\right )
\pi_{z_1}
\exp(Q(\rho \tau - t))_{z_1,z_2}
\exp(Q(t))_{z_2,x}
\exp(Q(t))_{z_2,y}
\ dt \\
%
%
%
%
&= 
 \displaystyle \sum_{z_1,z_2} 
\displaystyle \int_0^{\tau} \!
\left(
\dfrac{1}{2 N}
\right )
\exp \left ( 
\dfrac{- T}{2 N}
\right )
\pi_{z_1}
\exp(\rho Q(\tau -  T))_{z_1,z_2}
\exp(\rho Q(T))_{z_2,x}
\exp(\rho Q(T))_{z_2,y}
\ dT.
\end{align*}

Since the last line is equal to (2), 
the distribution at the leaves remains unchanged and 
the probability of coalescence also remains unchanged.
Generalizing this example, we can obtain the distribution from a species tree with scaled rate matrices by adjusting population size and branch lengths across the tree.

\section{Identifiability of the Modified Coalescent}
\label{sec: Identifiability}

One of the most fundamental concepts in model-based reconstruction is that of identifiability. A model parameter
is identifiable if any probability distribution arising
from the model uniquely 
determines the value of that parameter. 
For the purposes of phylogenetic reconstruction, it is
particularly important that the tree parameter of the model 
be identifiable in order to make consistent inference.
 
In the following paragraphs, we will use the notation 
$\mathcal{C}_{n}$ for the 
set of distributions obtained by varying the $n$-leaf
 tree parameter in the standard multispecies coalescent 
model, though the discussion 
applies equally to 
$\mathcal{C}_n^*$, 
$\mathcal{C}_n(N)$, and
$\mathcal{C}_n^*(N)$. 
To uniquely recover the 
unrooted tree parameter of the $n$-leaf multispecies
coalescent model 
we would require that for all $n$-leaf 
rooted trees $S_1$ and $S_2$ that are 
topologically distinct when the root vertex
of each is suppressed,
$\mathcal{C}({S_1}) \cap \mathcal{C}(S_2) = \emptyset$.
This notion of identifiability is unobtainable in most
instances and much stronger than is required in practice. 
Instead we often wish to establish
\emph{generic identifiability}. 
A model parameter is generically identifiable if the 
set of parameters from which the original parameter
cannot be recovered is a set of Lebesgue measure zero
in the parameter space. 
In our case, although we cannot
guarantee that 
$\mathcal{C}(S_1) \cap \mathcal{C}(S_2) = \emptyset$, 
we will show that that if we select parameters for either model
the resulting distribution will lie in 
$\mathcal{C}(S_1) \cap \mathcal{C}(S_2)$ 
with probability zero.

In \cite{chifmankubatko2015}, it was shown that the unrooted topology of the tree parameter for the coalescent model is generically 
identifiable using the machinery of analytic varieties.
Let
$$\psi_{S}: \Theta_{S} \mapsto \Delta^{\kappa^4 -1}$$
be the map from the continuous parameter space for $S$
to the probability simplex
with $\text{Im}(\psi_{S}) = \mathcal{C}_{S}$. 
Label the states of the model by the natural numbers 
$\{1, \ldots, \kappa\}$.
Given any two species tree parameters $S_1$ and $S_2$, the strategy
is to find a polynomial
$$f \in \mathbb{R}[q_{i_1\ldots i_n}: 1 \leq i_1, \ldots, i_n \leq \kappa]$$
such that for all 
$p_1 \in \mathcal{C}(S_1)$,
$f(p_1) = 0$, but for which there exists
$p_2 \in \mathcal{C}(S_2)$ such that
$f(p_2) \not = 0$. 
If we can show that the function
$$f \circ \psi_{S_2}: \Theta_{S_2} \to \mathbb{R}$$
is analytic, the theory of real analytic varieties 
and the existence of $p_2$ then imply that 
the zero set of $f \circ \psi_{S_2}$ 
is a set of Lebesgue measure zero 
\cite{Mityagin2015}.
Doing this for all pairs of $n$-leaf trees establishes the 
generic identifiability of the unrooted tree parameter of
$\mathcal{C}_n$.

We will show that the modified models introduced above are still
identifiable using the same approach. 
In the discussion proceeding \cite[Corollary 5.5]{chifmankubatko2015}, 
it was shown that for the 
standard multispecies coalescent, 
to establish identifiability of the tree parameter of 
the coalescent model 
for trees with any number of leaves it is enough to prove
the identifiability of the tree parameter for the
4-leaf model. Essentially the same proof of that theorem applies to the clockless coalescent giving us the following proposition.

\begin{op}
\label{4leavesisenough}
If the topology of the unrooted species tree parameter of $\mathcal{C}^*_4$ 
is generically 
identifiable then the topology of the unrooted species tree parameter of $\mathcal{C}^*_n$
 is generically identifiable for all $n$.
\end{op}

A similar proposition holds for the $p$-coalescent but a slight
modification is required.
The subtlety is illustrated in 
Figure \ref{fig: RestrictedSpeciesTrees} where a species
tree and its restriction to a 4-leaf subset of the leaves 
are shown. 
Notice that on the restricted tree, the effective population
size may now vary within a single branch. Therefore, 
to show the identifiability of the unrooted species
tree parameter of the $p$-coalescent for $n$-leaf trees, 
we must show the identifiability of the unrooted species tree parameter of a model on 4-leaf trees
that allows for a finite number of bands on each branch
with separate effective population sizes. 
We will revisit this point after the
proof of Theorem \ref{thm: identifiability16}, though it turns out to be 
rather inconsequential.

\begin{figure}[h]
\caption{A 5-leaf species tree $S$ with multiple effective population size parameters and its restriction to the
4-leaf subtree $S_{|\{a,b,d,e\}}$.
The image of the marginalization map applied to the model for 
$S$ will be the model for $S_{|\{a,b,d,e\}}$ 
with different effective population size parameters on different portions of $e_{ab}$.
}
\label{fig: RestrictedSpeciesTrees}
\begin{center}
\includegraphics[width=12cm]{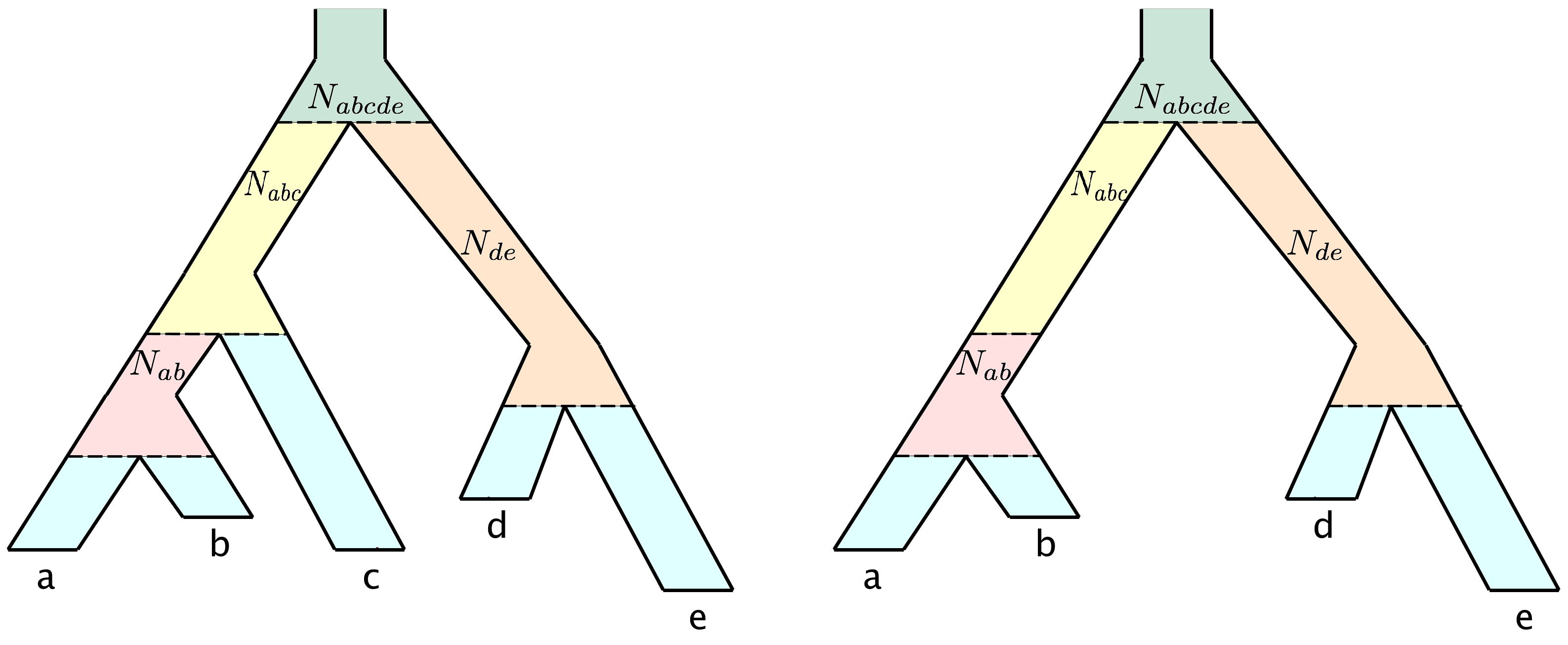}
\end{center}
\end{figure}

\subsection{The analyticity of $\psi_S$}
\label{sec: analyticity}

As mentioned above, to apply the results for 
analytic varieties, we need each of the
parameterization functions $\psi_S$ to be a real analytic
function in the continuous parameters 
of the model. 
That this is so may seem obvious to some
and was stated without proof in \cite{chifmankubatko2015}.
However, the coordinates of $\psi_S$ are functions
of the parameters defined by improper integrals.
We note that it is not true in general that if 
$f(\alpha,t)$ is a real analytic function that 
$F(\alpha) = \int_0^\infty f(\alpha,t) \ dt$ 
is real analytic -- even if $F(\alpha)$ is defined for all $\alpha$
(e.g., if $f(\alpha,t) = \frac{d}{dt}(\alpha \tanh(\alpha t))$, then $F(\alpha) = |\alpha|$). 
Further complicating matters is the fact
that the entries of the transition matrices are defined by
a matrix exponential. Thus, while they can be expressed
as convergent power series in the rate matrix parameters,
it is difficult to work with improper integral over these series.
For the models JC69, K2P, K3P, F81, HKY85, and TN93 these issues become 
irrelevant, as we can 
diagonalize the rate matrices and obtain a closed form
expression for the entries of the transition matrices. 
The entries are then seen to be exponential functions of
branch length and we can solve the improper integrals from 
the multispecies coalescent and obtain
exact formulas for each coordinate of $\psi_S$ that are
clearly analytic. Thus, we have the following proposition.

\begin{op}
\label{prop: analytic}
Let $S$ be a rooted $4$-leaf species tree. 
The parameterization map $\psi_S$ is analytic when the underlying nucleotide substitution model is any of
JC69, K2P, K3P, F81, HKY85, or TN93.
\end{op}

The rate matrix for the 4-state general time-reversible model is similar to a real symmetric matrix and is thus also diagonalizable.
However, actually writing down a closed form for the entries
of the transition matrix is not possible due to the large number
of computations involved. Instead, we will explain the steps
necessary to diagonalize the GTR rate matrix. 
While we will not be able to write the entries explicitly, we will
see that they must be expressions involving only elementary operations of the rate matrix parameters, roots of the rate matrix parameters, and exponential functions. This enables us to 
prove the following proposition about $\psi_S$.

\begin{op}
\label{prop: GTRanalytic}
Let $S$ be a rooted $4$-leaf species tree. 
Let $\psi_S$ be the parameterization map for the multispecies
coalescent model
when the underlying nucleotide substitution model is
the 4-state GTR model.
For a generic choice of continuous parameters $\boldsymbol{\theta}$,
there exists a neighborhood around $\boldsymbol{\theta}$ on which
$\psi_S$ is a real analytic function.
 \end{op}

\begin{proof}
Let $Q$ be a rate matrix for the 4-state GTR model. Then 
$$A = \text{diag}(\boldsymbol{\pi}^{1/2})Q\text{diag}(\boldsymbol{\pi}^{-1/2})$$
 is a real symmetric matrix and hence there exists
 a real orthogonal matrix $U$ and a diagonal matrix $D$ 
 such that $A = UDU^T.$
 As similar matrices, $Q$ and $A$ have the same eigenvalues
 and so  $\lambda_1 = 0$ is an eigenvalue of $A$. 
 Factoring the degree four characteristic equation of $A$, we can  use the cubic formula to write the 
other eigenvalues $\lambda_2, \lambda_3,$ and $\lambda_4$ 
explicitly in terms of the rate matrix parameters.
All eigenvalues of $A$ will be real and less than or equal to zero 
\cite[Section 2.2]{allmanetal2008}.

Suppose now that all eigenvalues of $A$ are distinct.
By the Cayley-Hamilton theorem, we have
$$A(A - \lambda_2I)(A - \lambda_3I)(A - \lambda_4I) = 0.$$ 
Since all eigenvalues of $A$ are distinct, the minimal and characteristic polynomials of $A$ are equal. Thus,
there exists a non-zero
column of $\prod_{i \not = j}(A - \lambda_jI)$ that
is an eigenvector of $A$ with eigenvalue $\lambda_i$. 
Generically, the first column of each of these matrices
will be non-zero and
so we can use the normalized first columns to construct the
matrix $U$ so that 
$Q = (\text{diag}(\boldsymbol{\pi}^{-1/2})U)
\text{diag}(0, \lambda_2, \lambda_3, \lambda_4)
(U^{T}\text{diag}(\boldsymbol{\pi}^{1/2}))$.

Therefore, 
$$P(t) = e^{Qt} = 
(\text{diag}(\boldsymbol{\pi}^{-1/2})U)
\text{diag}(1, e^{\lambda_2t}, e^{\lambda_3t},
e^{\lambda_4t})
(U^{T}\text{diag}(\boldsymbol{\pi}^{1/2})),$$
and each entry of the matrix exponential can be written as 
\begin{align}
P_{ij}(t) = \displaystyle \sum_{1 \leq k \leq 4}  f^{(ij)}_k(q)e^{\lambda_kt},
\end{align}
where the $f^{(ij)}_k(q)$ are rational functions of the rate matrix parameters and roots of the rate
matrix parameters coming from the cubic formula.
Thus, for a generic choice of parameters, the functions $p^*_{i_1\ldots i_n|(G,\mathbf{t})}$ are all
sums of products of these functions which are exponential  
in the branch length $t$. The formulas coming from the coalsecent
process, $f_h((G,\mathbf{t})|(S,\boldsymbol{\tau}))$, are also exponential functions in $t$. 
Because each $\lambda_i$ is guaranteed to be less than or equal to zero, when we integrate with respect to branch length each of the integrals converges. Therefore, $\psi_S$ can be written in closed form as an expression involving rational functions of the model parameters, roots of the model parameters, and exponential functions of both of these.

The expression for $\psi_S$ given above only fails to be analytic at points where one of the vectors used to construct $U$ is actually the zero vector so that the formula for normalizing the columns involves dividing by zero. This includes points where the eigenvalues of the rate matrix are not distinct.
Since each of the rate matrices
is diagonalizable, if all of the eigenvalues are not distinct then 
at least one of the vectors we have constructed will certainly be zero \cite[\S 6.4, Theorem 6]{Hoffman1971}. 
Therefore, the formula given above, which we will denote 
$\phi_S$, agrees
with $\psi_S$ everywhere that it is defined. 
Our choice of rate matrix parameters from the K3P model in Theorem \ref{thm: identifiability16} demonstrates that the analytic functions for the entries of the columns of $U$ are not identically zero. 
Therefore, for a generic choice of parameters 
$\boldsymbol{\theta}$, $\phi_S$ agrees with $\psi_S$ on a 
neighborhood around $\boldsymbol{\theta}$ and is a real analytic function in that neighborhood.
\end{proof}

We can extend the result in Proposition \ref{prop: analytic}
to include the generalized $\kappa$-state Jukes-Cantor model
since the entries of the transition matrices are given by

\[ P_{ij}(t) = 
\begin{cases} 
      \frac{1}{\kappa} + \frac{\kappa - 1}{\kappa}e^{-\mu t} & 
      i = j, \\
       \frac{1}{\kappa} - \frac{1}{\kappa}e^{-\mu t} & 
      i \not = j. \\
   \end{cases}
   \]

For the $\kappa$-state GTR model, because zero is an eigenvalue of $Q$, finding
all the eigenvalues of $Q$ amounts to finding the roots of a 
degree $(\kappa - 1)$ polynomial with coefficients that 
are polynomials of the rate matrix parameters. 
Thus, the same arguments also extend to any $\kappa$-state time-reversible model with $\kappa \leq 5$. 
Since there is no general algebraic solution for 
finding the roots of polynomials of degree five or higher (Abel-Ruffini Theorem), 
for $\kappa > 5$, a different argument is required. 
We do not pursue these arguments here as the $4$-state model
is by far the most commonly used in phylogenetics.

%
%
%
%
%
%
%
%

\subsection{Identifiability of the modified multispecies
coalescent for 4-leaf trees}

We may encode the probability distribution associated to a $\kappa$-state phylogenetic model on an $n$-leaf tree as an $n$-dimensional $\kappa \times \ldots \times \kappa$ tensor $P$
 where the entry $P_{i_1 \ldots i_n}$ is the probability of
observing the state $i_1 \ldots i_n$. 
In \cite[Section 4]{chifmankubatko2015}, the authors explain how to construct tensor flattenings according to a split of the species tree.
Our first result is the analogue of \cite[Theorem 5.1]{chifmankubatko2015}
for the modified coalescent models. We use the notation
$P_{(S,\boldsymbol{\tau},\boldsymbol{\theta})}$ to denote the probability tensor
that results from choosing a species tree $S$ with vector
of edge lengths $\boldsymbol{\tau}$ and continuous 
parameters $\boldsymbol{\theta}$.

\begin{thm} 
\label{thm: identifiability10}
Let $S$ be a four-taxon symmetric $((a,b),(c,d))$ or 
asymmetric $(a,(b,(c,d))$ species tree with a cherry $(c,d)$. 
Let $L_1|L_2$ be one of the splits $ab|cd$, $ac|bd$, or $ad|bc$.
Consider the clockless coalescent 
when the underlying nucleotide substitution model is any of the
following: JC69, K2P, K3P, F81, HKY85, TN93, or GTR.
If $L_1|L_2$ is a valid split for $S$, then for all 
$P_{(S,\boldsymbol{\tau},\boldsymbol{\theta})} \in  
\mathcal{C}^*(S)$,
$$rank(Flat_{L_1|L_2}(P_{(S,\boldsymbol{\tau},\boldsymbol{\theta})})) \leq 10.$$
\end{thm}

\begin{proof}
Let $L_1|L_2$ be the split $ab|cd$ that is valid for $S$
 and consider the distribution
$P_{(S,\boldsymbol{\tau},\boldsymbol{\theta})}$. 
 Increase the length
of one of the leaf edges in the cherry of $S$ to create
a new vector of edge lengths $\boldsymbol{\xi}$. 
Without loss of generality,
suppose $c$ is extended so that  
$\xi_c > \tau_c$ as in Figure \ref{fig: ExtendedSpeciesTrees}.

First, we claim that 
$$
rank(Flat_{L_1|L_2}(P_{(S,\boldsymbol{\xi},\boldsymbol{\theta})}))
\leq
rank(Flat_{L_1|L_2}(P_{(S,\boldsymbol{\tau},\boldsymbol{\theta})}))
$$
Notice that since coalescent events do not 
happen in the leaves of the species tree, the 
gene tree histories and the 
formulas for the 
gene tree distributions from $(S,\boldsymbol{\tau})$ and $(S,\boldsymbol{\xi})$
are identical.
The only difference is that the leaf edge labeled
by $c$ in the gene tree $(G,\mathbf{t})$ from $(S,\boldsymbol{\xi})$ is
longer by $\xi_c - \tau_c$  
than the same edge in the gene tree $(G,\mathbf{t})$ from $(S,\boldsymbol{\tau})$.
The probability of observing the state 
$i_1i_2i_3 i_4$ from these two gene trees will not be the same, 
so let us express this probability as 
$p^*_{i_1i_2i_3 i_4|(G,\textbf{t},\boldsymbol{\theta})}$ when the species tree is
 $(S,\boldsymbol{\tau})$ and
$q^*_{i_1i_2i_3 i_4|(G,\textbf{t},\boldsymbol{\theta})}$ when the 
species tree is $(S,\boldsymbol{\xi})$.
Extending the branch of a gene tree is equivalent to 
grafting a new edge onto the leaf edge to create an internal
 vertex of degree two. 
 To compute the probability of observing a particular
state at the leaves of the extended gene tree, 
we sum over all possible states of this vertex. 
For clarity of notation, let us represent the matrix of
transition probabilities along the grafted edge by
$M = e^{Q(\xi_c - \tau_c)}$.
Thus, 

$$
q^*_{i_1i_2i_3 i_4|(G,\textbf{t},\boldsymbol{\theta})}=
\displaystyle 
\sum_{1 \leq j \leq 4}
(M_{ji_3}) p^*_{i_1i_2ji_4|(G,\textbf{t},\boldsymbol{\theta})}.
$$

Therefore, the total probability for a particular history is given by

\begin{align*}
p_{i_1i_2i_3i_4|h,(S,\boldsymbol{\xi},\boldsymbol{\theta})}&=
\displaystyle
\int_{\mathbf{t}}
q^*_{i_1i_2i_3i_4|(G,\textbf{t},\boldsymbol{\theta})} 
f_h((G,\mathbf{t})|(S,\boldsymbol{\xi}))
\ d\mathbf{t} \\
&=
\displaystyle
\int_{\mathbf{t}}
\left (
 \displaystyle 
\sum_{1 \leq j \leq 4}
(M_{ji_3})
p^*_{i_1i_2ji_4|(G,\textbf{t},\boldsymbol{\theta})} 
\right )
f_h((G,\mathbf{t})|(S,\boldsymbol{\tau}))
\ d\mathbf{t} \\
&=
 \displaystyle 
\sum_{1 \leq j \leq 4}
(M_{ji_3})
\displaystyle
\left (
\int_{\mathbf{t}}
 p^*_{i_1i_2ji_4|(G,\textbf{t},\boldsymbol{\theta})} 
f_h((G,\mathbf{t})|(S,\boldsymbol{\tau}))
\ d\mathbf{t}
\right ) \\
&=
 \displaystyle 
\sum_{1 \leq j \leq 4}
\displaystyle
(M_{ji_3})
p_{i_1i_2j i_4|h,(S,\boldsymbol{\tau},\boldsymbol{\theta})} \\
\end{align*}

Summing over all histories, we also obtain

$$ 
p_{i_1i_2i_3i_4|(S,\boldsymbol{\xi},\boldsymbol{\theta})}=
\displaystyle 
\sum_{1 \leq j \leq 4}
\displaystyle
(M_{ji_3})
p_{i_1i_2j i_4|(S,\boldsymbol{\tau},\boldsymbol{\theta})}.$$

Now consider the column of
$Flat_{L_1|L_2}(P_{(S,\boldsymbol{\xi},\boldsymbol{\theta})})$
indexed by $i_3i_4$. 
The formula above shows that this column is 
a linear combination of the columns of
$Flat_{L_1|L_2}(P_{(S,\boldsymbol{\tau},\boldsymbol{\theta})})$
 indexed by
$1i_4,2i_4, 3i_4,$ and $4i_4$.
Therefore, 
$$
rank(Flat_{L_1|L_2}(P_{(S,\boldsymbol{\xi},\boldsymbol{\theta})}))
\leq
rank(Flat_{L_1|L_2}(P_{(S,\boldsymbol{\tau},\boldsymbol{\theta})})).
$$
Given a 4-leaf species tree $(S,\boldsymbol{\tau})$ 
with a $(c,d)$ cherry, consider the
tree $(S,\boldsymbol{\tau}^0)$, where all entries of $\boldsymbol{\tau}^0$ are the same
as for $\boldsymbol{\tau}$ except that $\tau^0_c = \tau^0_d = 0$.
The tree $(S,\boldsymbol{\tau}^0)$ may not be equidistant, but it is still
clear from the symmetry in the cherry that 
for any choice of continuous parameters we will have
$$
p_{i_1i_2i_3i_4|(S,\boldsymbol{\tau}^0,\boldsymbol{\theta})}=
p_{i_1i_2i_4i_3|(S,\boldsymbol{\tau}^0,\boldsymbol{\theta})},
$$ which implies that
$
rank(Flat_{L_1|L_2}(P_{(S,\boldsymbol{\tau}^0,\boldsymbol{\theta})}))
\leq 10.
$
Now we can extend each of the cherry edges  $c$ and $d$ in
$(S,\boldsymbol{\tau}^0)$ to construct $(S,\boldsymbol{\tau})$. Since we have shown that
extending these branches does not change the rank of
the flattening, the result follows.

\begin{thm} 
\label{thm: identifiability16}
Let $S$ be a four-taxon symmetric $((a,b),(c,d))$ or 
asymmetric $(a,(b,(c,d))$ species tree with a cherry $(c,d)$. 
Let $L_1|L_2$ be one of the splits $ab|cd$, $ac|bd$, or $ad|bc$.
Consider the clockless coalescent 
when the underlying nucleotide substitution model is any of the
following: JC69, K2P, K3P, F81, HKY85, TN93, or GTR.
 If $L_1|L_2$ is not a valid split for $S$, then for generic
distributions $P_{(S,\boldsymbol{\tau},\boldsymbol{\theta})} \in  
\mathcal{C}^*(S)$,
$$rank(Flat_{L_1|L_2}(P_{(S,\boldsymbol{\tau},\boldsymbol{\theta})}))  =16.$$
\end{thm}

We verify the statement computationally by finding a specific
choice of parameters for the symmetric tree
where 
$\det(Flat_{ac|bd}(P_{(S,\boldsymbol{\tau},\boldsymbol{\theta})})) \not = 0$ and 
$\det(Flat_{ad|bc}(P_{(S,\boldsymbol{\tau},\boldsymbol{\theta})}))\not = 0$
and doing the same for the 
asymmetric tree. 
In fact, we can address both cases with one tree by
letting $S$ be the symmetric tree and setting 
$\tau_{ab} = 0$. 
There is now only one internal edge and we choose
$\tau_c = \tau_d=\tau_{cd} = 1$, $\tau_{a} = \tau_{b}= 2$ and $N_{cd}=1$. We also need only look at one of the 
flattening matrices, since as observed in \cite[Theorem 5.1]{chifmankubatko2015}, the flattening matrices for the invalid splits
will be the same.

We choose continuous parameters
from the Jukes-Cantor model by setting the off-diagonal
entries of the rate matrix to $0.025$ and letting 
the root distribution be uniform. With these choices, both flattening matrices for the invalid splits are rank 16. Since the Jukes-Cantor
model is contained in JC69, K2P, K3P, F81, HKY85, and TN93, this choice of parameters establishes the result for each of these.

The Jukes-Cantor model is of course also contained in the 4-state GTR model. However, recall from Proposition \ref{prop: GTRanalytic} that we were not able to verify that each coordinate
of $\psi_S$ is a real analytic function on the entire parameter space. 
Instead, in Proposition \ref{prop: GTRanalytic}, we described the function $\phi_S$ that agrees with $\psi_S$ everywhere it is defined, which is to say, almost everywhere.
Therefore, we will take a slightly different tack in proving the result for the GTR model. As before, we only need to look at either invalid split, so denote by $d_{ac|bd}(\phi_S)$ the determinant of the flattening matrix where the formulas for 
$\phi_S$ are used in place of the variables $p_{i_1i_2i_3i_4}$.  
Since they agree almost everywhere, if 
$\det(Flat_{ac|bd}(P))$ is zero on a set of positive measure in the parameter space, the same is true of $d_{ac|bd}(\phi_S)$. 
Recall that when constructing the coordinate functions of
$\phi_S$, we normalize the vectors of $U$, and that the
resulting coordinate functions fail to be defined only when one of
these four vectors is identically zero.
Therefore, it is possible to multiply $d_{ac|bd}(\phi_S)$ by these vector norms to clear denominators and 
obtain the function $D_{ac|bd}(\phi_S)$ that is analytic on the entire parameter space.
Importantly, $D_{ac|bd}(\phi_S)$ is zero anywhere
that $d_{ac|bd}(\phi_S)$ is. 
Therefore, if $\det(Flat_{ac|bd}(P))$ is zero on a set
of positive measure, so is $d_{ac|bd}(\phi_S)$, and hence so is $D_{ac|bd}(\phi_S)$. But since it is analytic, this would imply that $D_{ac|bd}(\phi_S)$ is identically zero.

Therefore, we will establish the result by verifying that 
$D_{ac|bd}(\phi_S)$ is not identically zero. 
Choosing parameters from the K3P model with
$a = 1$, $b = 2$ and $c = 3$ (Figure \ref{fig: rate matrices}),
the rate matrix has 
eigenvalues $0, -6, -8,$ and $-10$ and the matrix of normalized 
eigenvectors from Proposition \ref{prop: GTRanalytic} is 
$$U = \dfrac{1}{2} 
\begin{pmatrix}
        1  & -1  & 1  & 1 \\
       1  &  1   & -1  & -1\\
       1  &  1  &  1  & - 1  \\
       1  &  -1 &  -1  & 1  \\
        \end{pmatrix}.$$
Specifying the other
parameters of the model as above, we find that 
$d_{ac|bd}(\phi_S)$ is not identically zero, and 
since none of the columns of $U$ are zero, neither
is $D_{ac|bd}(\phi_S)$.
\end{proof}

\begin{figure}[h]
\caption{Extending one leaf in a cherry of $S$.}
\label{fig: ExtendedSpeciesTrees}
\begin{center}
\includegraphics[width=6cm]{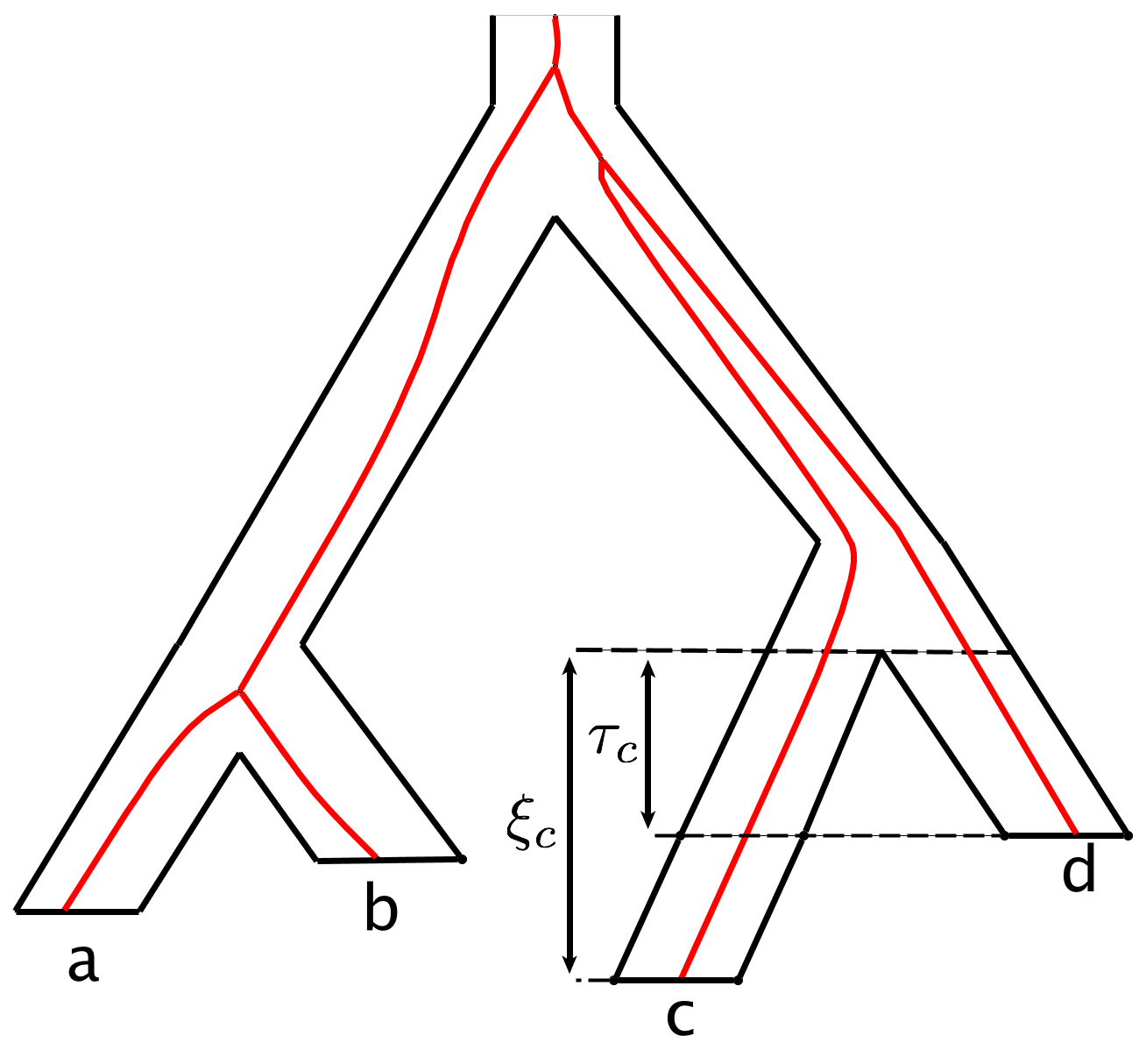}
\end{center}
\end{figure}

A matrix has rank less than 
or equal to $k$ if and only if all of its $(k+1) \times (k+1)$ 
minors vanish.
Therefore, by Theorem \ref{thm: identifiability10},
if $S_1$ is a 4-leaf tree displaying the split $ab|cd$ then
all of the $11 \times 11$ minors of the flattening matrix for 
that split vanish on the set $\mathcal{C}^*(S_1)$.
Theorem \ref{thm: identifiability16} shows
that if $S_2$ is a $4$-leaf tree that does not display the split 
$ab|cd$, then one of these minors will not vanish on the set
$\mathcal{C}^*(S_2)$. Hence, as per the discussion at the 
beginning of Section \ref{sec: Identifiability}, the set of parameters
for $S_2$ mapping into  
$\mathcal{C}^*(S_1) \cap \mathcal{C}^*(S_2)$ is a set 
of measure zero. Thus, the unrooted species tree parameter of the clockless coalescent is generically identifiable.

Since coalescent events do not occur in leaves of the species tree, Theorem \ref{thm: identifiability10} applies equally to the $p$-coalescent and clockless $p$-coalescent models. 
Following Proposition \ref{4leavesisenough}, we observed
that showing the identifiability of the unrooted species tree parameter of the $p$-coalescent requires proving the same result for 4-leaf trees in a model that allows
multiple effective population size parameters on a single edge. Specifically, a model where
the maximum number of effective population size parameters on each edge is bounded by the total number of edges in an $n$-leaf binary tree, $2n -3$, is sufficient to prove
identifiability of the unrooted species tree parameter in $\mathcal{C}^*_n(N)$.
We have essentially already proven this for all $n$.

Both distributions in the 
proof of Theorem \ref{thm: identifiability16} are still contained in the model where we allow
multiple effective population size parameters on each edge
since we can just choose the multiple parameters 
on $e_{cd}$ to all be equal to one. 
We must still verify that the 
parameterization map for this model is analytic, 
but the argument from Section \ref{sec: analyticity}
remains unchanged when we allow multiple
effective population size parameters on each edge.
Thus, the same choices of parameters from the proof of 
Theorem \ref{thm: identifiability16} establish the result for the clockless $p$-coalescent. 
We also intentionally chose a point corresponding to an
equidistant tree so that it applies to the $p$-coalescent.
Thus, we have the following corollary.

\begin{cor}
The unrooted species tree parameters of 
the clockless coalescent, the $p$-coalescent,
and the clockless $p$-coalescent models on an 
$n$-leaf tree 
are generically
identifiable for all $n$ 
when the underlying nucleotide substitution model is any of the
following: JC69, K2P, K3P, F81, HKY85, TN93, or GTR.
\end{cor}

\subsection{Identifiability with invariant sites and gamma distributed rates}

It is well known that the rate of evolution may vary across sites \cite{yang1993,yang1994}.
One way to account for this is to let each site evolve according to the same model but where the rate matrix at each site is scaled
by a random factor drawn from a specified
gamma distribution. 
If the underlying nucleotide substitution model is assumed to be the GTR model, this is what is known 
as the GTR+$\Gamma$ model.

In practice, the gamma distribution is approximated using
$m$ rate categories, each with probability $\frac{1}{m}$.  
Let the density of the gamma distribution
$G(\alpha, \beta)$ be given by $g(\rho;\alpha, \beta)$.
The rate for each category $\rho_i \in (0 ,\infty)$ for 
$1 \leq i \leq m$ is determined by first setting
$a_1 = 0$ and then computing the 
points $a_i$ for $2 \leq i \leq {m}$ such that
$$\int_{a_i}^{a_{i+1}} g(\rho,\alpha, \beta) \ d\rho = \frac{1}{m}.$$
For $1 \leq i \leq m-1$, the rate $\rho_i$ is then given by the mean over each interval,
$$
\rho_i =
m \cdot
 \int_{a_i}^{a_{i+1}} \rho \cdot g(\rho,\alpha, \beta) \ d\rho = 
\left ( \frac{m \alpha}{\beta} \right )
(I(b\beta, \alpha + 1) - 
I(\alpha\beta, \alpha + 1))
$$
where 
$$I(z,\alpha) = 
\left(
\frac{1}{\Gamma(\alpha)}
\right)
\int_0^z \exp(-x) x^{\alpha - 1} \ dx.
$$ 
The last rate $\rho_m$ is then
$\left ( \frac{m \alpha}{\beta} \right) - \sum_{i=1}^{m-1}\rho_i$.
We reproduce these formulas from \cite{yang1994} 
here only to show that the rates can be
expressed as analytic functions in the parameters 
$\alpha$ and $\beta$ and consequently that 
the distributions from the GTR+$\Gamma$ model
are given by real analytic functions of the parameters.


It is also common to account for invariant sites by
using the GTR+I+$\Gamma$ model, where
$\delta$ is the proportion 
of sites that do not evolve at all. 
The multispecies coalescent with the $m$-discrete $\kappa$-state GTR+I+$\Gamma$ model was shown to exhibit the same flattening ranks as the multispecies coalescent with the 
$\kappa$-state GTR model in \cite{chifmankubatko2015}.
This is not terribly surprising as a probability 
distribution from the former is the sum of 
$m+1$ distributions each satisfying the same linear 
relations.
Explicitly, 

$$p^{\Gamma + \delta}_{i_1i_2i_3i_4|(S, \tau, \boldsymbol{\theta})} 
=
\frac{(1 - \delta)}{m}
(p^{\rho_1}_{i_1i_2i_3i_4|(S, \tau, \boldsymbol{\theta})} + \ldots + 
p^{\rho_m}_{i_1i_2i_3i_4|(S, \tau, \boldsymbol{\theta})}) + 
\delta(z_{i_1i_2i_3i_4|\boldsymbol{\theta}}),
$$
where 
$p^{\rho_j}_{i_1i_2i_3i_4|(S, \tau, \boldsymbol{\theta})}$
is the probability of observing $i_1i_2i_3i_4$ from 
the multispecies coalescent model with scaling factor $\rho_j$ and 
$z_{i_1i_2i_3i_4|\boldsymbol{\theta}}$ 
is the probability of observing this state an an invariant site.
If $S$ has a $(c,d)$ cherry as above, then 
each summand is contained in the
linear space defined by the linear relations of the form
$p_{\star\star i_3i_4} - p_{\star\star i_4i_3}$
in the distribution space. 
The sum satisfies these relations
as well, so we have
$$
rank(Flat_{ab|cd}(P^{\Gamma + \delta}_{(S, \boldsymbol{\tau}, \boldsymbol{\theta})}))
\leq { \kappa + 1 \choose 2}.
$$

For a non-equidistant tree, the same result
no longer applies. 
If we extend a leaf in the cherry of $S$ as we did in Theorem
\ref{thm: identifiability10}, we will still have 
$$rank(Flat_{ab|cd}(P^{\rho_j}_{(S, \boldsymbol{\tau}, \boldsymbol{\theta})}))
\leq  { \kappa + 1 \choose 2},$$
but the particular linear relationships satisfied by the columns
of each flattening matrix will depend on
the entries of the transition matrix on the extended edge, which
in turn depend on the $\rho_i$. 
However, we can obtain an analogous result for 
JC+I+$\Gamma$, where the JC refers to the 
$\kappa$-state Jukes-Cantor model. 
When $\kappa = 4$, we prove the result
for $m = 2,3,$ and $4$, as four is the most common
number of categories used in actual phylogenetic
applications \cite{liogoldman1998}.

 \begin{thm} 
\label{thm: gamma}
Let $S$ be a four-taxon symmetric $((a,b),(c,d))$ or 
asymmetric $(a,(b,(c,d))$ species tree with a cherry $(c,d)$. 
Let $L_1|L_2$ be one of the splits $ab|cd$, $ac|bd$, or $ad|bc$.
For $\kappa \geq 4$, consider the $\kappa$-state $m$-discrete 
JC+I+$\Gamma$ 
model under the coalescent with species tree $S$ and $m \leq 4$.

\begin{enumerate}
\item If $L_1|L_2$ is a valid split for $S$, then for all 
$P^{\Gamma+I}_{(S,\boldsymbol{\tau},\boldsymbol{\theta})} \in   \mathcal{C}^*(S,N)$,
$$rank(Flat_{L_1|L_2}(P_{(S,\boldsymbol{\tau},\boldsymbol{\theta})}^{\Gamma+I})) 
\leq \kappa^2 - {\kappa - 1 \choose 2}$$
\item If $L_1|L_2$ is not a valid split for $S$, 
then for generic distributions 
$P_{(S,\boldsymbol{\tau})}^{\Gamma+I} \in \mathcal{C}^*(S,N)$,
$$rank(Flat_{L_1|L_2}(P_{(S,\boldsymbol{\tau},\boldsymbol{\theta})}^{\Gamma+I})  )
> \kappa^2 - {\kappa - 1 \choose 2}$$
\end{enumerate}
\end{thm}

\begin{proof}
Let $L_1|L_2$ be the split $ab|cd$ that is valid for $S$
 and consider the distribution
$P_{(S,\boldsymbol{\xi},\boldsymbol{\theta})}$ from the Jukes-Cantor model.
Without loss of generality,
suppose $\xi_c \geq \xi_d$.
Construct the vector $\tau$ with all entries equal
to those of $\boldsymbol{\xi}$ but with 
$\tau_c = \xi_d$. 
Again, by symmetry, we have
$$p_{\star\star i_3 i_4|(S,\boldsymbol{\tau},\boldsymbol{\theta})} =
p_{\star\star i_4 i_3|(S,\boldsymbol{\tau},\boldsymbol{\theta})}.
$$
As in Theorem \ref{thm: identifiability10}, we will identify the
tree $(S,\boldsymbol{\xi})$ as an extension of $(S,\boldsymbol{\tau})$.
For the JC model, there are only two distinct 
entries of $M = e^{Q(\xi_c - \xi_d)}$.
Let $M_{ij} = a$ if $i = j$ and $b$ otherwise.
Therefore, we have
$$ 
p_{\star\star i_3i_4|(S,\boldsymbol{\xi},\boldsymbol{\theta})}=
a p_{\star \star i_3 i_4|(S,\boldsymbol{\tau},\boldsymbol{\theta})} + 
\displaystyle 
\sum_{i_3 \not = j}
\displaystyle
b p_{\star \star j i_4|(S,\boldsymbol{\tau},\boldsymbol{\theta})},$$ 
and one can check that for distinct
$k_1,k_2,k _3 \in [\kappa]$, the distribution 
 $P_{|(S, \xi, \boldsymbol{\theta})}$, satisfies
$$
p_{\star\star k_1k_2} -
p_{\star\star k_1k_3} -
p_{\star\star k_2k_1} + 
p_{\star\star k_2k_3} + 
p_{\star\star k_3k_1} -
p_{\star\star k_3k_2}=0.
$$
We obtain such a relation for any 3-element subset of $\kappa$.
Moreover, since this linear relation does not depend on
 $a$ or $b$, it is satisfied by
 $P^{\rho_i}_{|(S, \xi, \boldsymbol{\theta})}$,
by $Z_{|\boldsymbol{\theta}}$, 
and hence, by any distribution from the 
$m$-discrete JC+I+$\Gamma$ model.
Consider the ${\kappa - 1 \choose 2}$ relations that come
from choosing 3-element subsets of the form 
$\{k_1,k_2,\kappa\}$.
For all $k_1,k_2 \in [\kappa -1]$, 
exactly one of these relations involves the variable
$p_{\star\star k_1k_2}$.
Therefore, these relations are linearly independent, and so
(1) follows.

In \cite[Theorem 5.1]{chifmankubatko2015} the authors show that 
for all $m$, when $\kappa \geq 4$,
 if $L_1|L_2$ is
not a valid split for $S$, then 
$$
rank(Flat_{L_1|L_2}(P_{(S,\boldsymbol{\tau},\boldsymbol{\theta})}^{\Gamma+I})  
> \kappa^2 - \kappa.
$$
When $\kappa \geq 5 $, we have 
 \begin{align*}
 \kappa^2 - \kappa
 &\geq
\kappa^2 - {\kappa - 1 \choose 2}, \\
 \end{align*}
which establishes our result.
For $\kappa \leq 4$, we must produce a choice of parameters to prove that the claim holds for $m=2,3,$ and $4$ and for both the symmetric and asymmetric trees. 
Choosing  
$\alpha = \beta = 1$
and the same continuous JC69 parameters from 
Theorem \ref{thm: identifiability16} establishes the result.
 \end{proof}

Since all of the parameterization functions involved
are analytic, this is enough to prove the identifiability of the 
unrooted species tree parameter of the JC+I+$\Gamma$ model.
Thus, we have the following corollary.

\begin{cor}
The unrooted species tree parameters of 
the clockless coalescent, the $p$-coalescent,
and the clockless $p$-coalescent models on an 
$n$-leaf tree 
are generically
identifiable for all $n$ 
when the underlying nucleotide substitution model is the
$m$-discrete $\kappa$-state JC+I+$\Gamma$ model with
$\kappa \geq 5$ and $m \in \mathbb{N}$ and with $\kappa = 4$ and $m = 2,3,$ or $4$.
\end{cor}

Moreover, the parameters that we used to demonstrate
that the invalid flattenings are full rank come from 
an exponential distribution. Therefore, the same result holds 
for a model where the $m$ rates are constructed from 
an exponential distribution. In fact, this also applies to a more
general variable rates model where the $m$ rates are free
parameters.

\section{Conclusions
}

\label{sec: Conclusions}

In the previous section, we have proven that the unrooted
species tree parameters of several more generalized versions of the multispecies coalescent model are generically identifiable.
Moreover, the means by which we have proven 
identifiability give us the necessary framework for 
reconstructing the unrooted species tree from data. 
In each case,
we showed that we can reconstruct the unrooted quartets
of the species tree parameter if we know the distribution exactly 
by taking ranks of the flattening matrices. 
Specifically, for a 4-state model and generic choices of parameters, 
we showed that the rank of the flattening matrix
for the quartet compatible with
the species tree will be less than or equal to $10$ 
while the other two flattening matrices will both be rank 16. 

This gives a natural method for inferring the unrooted species
tree from actual biological data. Specifically, for each quartet, we infer the unrooted quartets of the species
tree by determining which of the three flattening matrices
 is closest to being rank 10. The method of singular value decomposition from linear algebra already provides a means of determining how close a matrix is to being of a certain rank \cite{golubvanloan2013}.
 This is exactly the procedure used by the method SVDQuartets, which is already fully implemented in the PAUP$^*$ software \cite{swofford2016}. Hence, there is strong theoretical justification for applying SVDQuartets for phylogenetic 
 reconstruction even when effective population sizes vary throughout the tree or when the molecular clock does not hold.

The model presented in Section 2, as well as that presented in Chifman and Kubatko (2015)\nocite{chifmankubatko2015}, describes the situation in which gene trees are randomly sampled under the multispecies coalescent model, and then sequence data for a single site evolve along each sampled gene tree according to one of the standard nucleotide substitution models. Data generated in this way have been termed ``coalescent independent sites'' \cite{tiankubatko2016} to distinguish them from SNP data. Although coalescent independent sites and SNP data refer to observations of single sites that are assumed to be conditionally independent samples from the model given the species tree, SNP data are generally biallelic, while coalescent independent sites may include three or four nucleotides at a site, or may be constant. Thus SNP data are a subset of coalescent independent sites data, and the results derived here apply to these data as well.  

The other situation in which one might wish to apply these results is to multilocus data. Multilocus data are data in which individual genes are sampled from the species tree under the multispecies coalescent, but for each sampled gene tree, many individual sites are observed.  Typical genes observed in phylogenomic studies range from 100 base pairs (bp) to 2,000bp in size, though most are less than 500bp. The site patterns observed within a gene are not independent observations under the model because they share the same gene tree, and thus it is not immediately obvious that the results presented here apply to this case. However, consider the case in which a large sample of genes, say $W$, is obtained, and for each gene, $S$ sites are observed. Then, the flattening matrices of site pattern counts constructed from such data will be $S$ times the flattening matrix of site pattern counts that would have been observed if only a single site had been observed from each gene tree, which does not change the matrix rank. It is clear that as $W \rightarrow \infty$,  the correct theoretical distribution will be well-approximated by the observed site pattern frequencies, and the results presented here will hold. In practice, the genes will vary in their lengths and a more careful argument is required. We have elsewhere carried out thorough simulation studies to show that the methods used in SVDQuartets hold for multilocus data as well as for SNP data and for coalescent independent sites for the original model \cite{chifmankubatko2014}, and we verify these results under the model presented here in comparison with other species tree estimation methods elsewhere \cite{longkubatko2017}.

 We note two possible criticisms of this method. The first is that, while we showed that generically the flattening matrices for the invalid splits will be rank 16, we have no theoretical guarantees that they are not arbitrarily close to the set of rank 10 matrices. 
 Therefore, we do not know a priori that this method will 
provide any insight with a finite amount of either simulated or actual biological data. Along the same lines, determining
 that a flattening matrix is close to the set of rank 10 matrices
 does not necessarily mean that it is close to the set of distributions arising from a coalescent model, as the latter is properly contained in the former.  While both valid considerations, 
they appear to be academic, as SVDQuartets has already
been shown to be an effective reconstruction method on 
several data sets, both real and simulated 
\cite{chifmankubatko2014,chouetal2015}.
As mentioned above, in a forthcoming paper, we will demonstrate that SVDQuartets
 also works well in practice by simulating data from these modified coalescent
models and applying the method to real biological data sets known to violate
the molecular clock.

In recent years, the amount of sequence data available for species tree inference has increased rapidly, presenting significant computational challenges for most model-based species tree inference methods that accommodate the coalescent process.   The SVDQuartets method, however, is fully model-based but inference using this method is much more computationally efficient than other methods.  This is because, for each quartet considered, all that is required is construction of the three flattening matrices, which involves the simple task of counting site patterns, and computation of singular values from these $16 \times 16$ matrices.  In addition, increases in sequence length benefit the performance of the method (because site pattern probabilities are estimated more accurately) with almost no increased computational cost, in direct contrast to other species tree inference methods, such as *BEAST \cite{heleddrummond2010} and SNAPP \cite{bryantetal2012}.  In the work presented here, we show that the theory underlying the SVDQuartets method holds in much more general settings than originally suggested. In particular, the method can be applied to data that violate the molecular clock and to the case in which each branch of the species has a distinct effective population size. Thus, this work is a significant advance that will contribute meaningfully to the collection of methods available to infer species-level phylogenies from phylogenomic data in very general settings.


%
\bibliography{bibfile}
\bibliographystyle{abbrv}

\end{document}